\documentclass[10pt,a4paper]{article}

\usepackage[T1]{fontenc}
\usepackage[utf8]{inputenc}
\usepackage{mathrsfs}       
\usepackage{calligra}      
\usepackage{verbatim}      

\usepackage{amsmath, amsfonts, amssymb, amscd, amsthm}
\usepackage[all]{xy}       

\usepackage[margin=2.5cm]{geometry} 

\usepackage{graphicx}
\usepackage{tikz}
\usepackage{float}          
\usepackage[figuresright]{rotating} 

\usepackage{longtable}      
\usepackage{multirow}       
\usepackage{makecell}       
\usepackage{diagbox}        
\usepackage{booktabs}       

\usepackage[perpage, symbol*]{footmisc} 
\usepackage[justification=centering]{caption} 
\usepackage[numbers, sort&compress]{natbib}    

\usepackage{hyperref}
\hypersetup{
    colorlinks=true,
    linkcolor=blue,
    urlcolor=blue,
    citecolor=blue,
    bookmarksnumbered=true
}

\newtheorem{theorem}{\textbf{Theorem}}
\newtheorem{lemma}{\textbf{Lemma}}
\newtheorem{definition}{\textbf{Definition}}
\newtheorem{remark}{\textbf{Remark}}
\newtheorem{corollary}{\textbf{Corollary}}

\newtheorem{example}{\textbf{Example}}

\newcommand{\F}{\mathbb{F}}

\begin{document}
\baselineskip 17pt
\title{\Large\bf  A Generic Construction on Self-orthogonal Algebraic Geometric Codes and Its Applications}
\author{\large  Puyin Wang \quad\quad Jinquan Luo*}\footnotetext{The authors are with School of Mathematics
and Statistics \& Hubei Key Laboratory of Mathematical Sciences, Central China Normal University, Wuhan China 430079.\\
 E-mail: p.wang98@qq.com(P.Wang), luojinquan@ccnu.edu.cn(J.Luo)}
\date{}
\maketitle
{\bf Abstract}:
    In this paper, we introduce a generalized sufficient condition for the self-orthogonality of AG codes, formulated in terms of residues and based on the algebraic structures of finite fields and the geometric properties of algebraic curves. We also present a generic construction of self-orthogonal AG codes from self-dual MDS codes. Using these approaches, we construct several families of self-dual and almost self-dual AG codes. These codes combine two merits: good performance as AG codes whose parameters are close to the Singleton bound, together with Euclidean (or Hermitian) self-dual/self-orthogonal property. Furthermore, some AG codes with Hermitian self-orthogonality can be applied to construct quantum codes with notably good parameters.


\section{Introduction}

\qquad Let $q$ be a power of a prime $p$, and let $\mathbb{F}_q$ denote the finite field with $q$ elements. An $[n,k,d]$ linear code $C$
over $\mathbb{F}_q$ is a $k$-dimensional subspace of $\mathbb{F}_q^n$ with minimum distance $d.$  For a linear code $C$, the Euclidean dual $C^\perp$ is defined as
$$C^\perp=\left\{(c_1^\prime,c_2^\prime,\ldots,c_n^\prime)\in \mathbb{F}_q^n\mid\sum_{i=1}^nc_i'c_i=0,\:\forall\:(c_1,c_2,\ldots,c_n)\in C\right\}.$$
The code $C$ is Euclidean self-orthogonal if it satisfies $C\subseteq C^\perp$. It is Euclidean self-dual if  $C=C^\perp$.
Similarly, for a linear code $C$ over $\mathbb{F}_{q^2}$,  the Hermitian dual $C^{\perp_H}$ is defined as
$$C^{\perp_H}=\left\{(c_1^\prime,c_2^\prime,\ldots,c_n^\prime)\in \mathbb{F}_{q^2}^n\mid\sum_{i=1}^n{c_i^\prime}^qc_i=0,\:\forall\:(c_1,c_2,\ldots,c_n)\in C\right\}.$$
The code $C$ is  Hermitian self-orthogonal if it satisfies $C\subseteq C^{\perp_H}$. It is  Hermitian self-dual if $C=C^{\perp_H}$.

Euclidean self-orthogonal codes not only possess special structures but also have a wide range of applications. For instance, they can be applied to construct quantum codes \cite{CRSS}. Moreover, self-dual codes, as a special type of self-orthogonal code, are of considerable theoretical interest. They can be applied to construct some $t$-designs, such as \cite{AM} and \cite{GB}. Besides, they are closely connected to lattices and modular forms \cite{BBH}, \cite{H}, \cite{Sv}.  Furthermore, self-dual codes have applications in secret sharing. In \cite{CDGU}, linear secret sharing schemes with specific access structure are constructed from self-dual codes.

In recent years, significant progress has been made in the construction of Euclidean self-orthogonal codes. For instance, \cite{ZKL} explores their construction through cyclic codes. In \cite{SWKS}, Euclidean self-orthogonal codes are constructed via non-unital rings, and \cite{ZLD} presents the construction of several Euclidean self-orthogonal codes in a generic way. In \cite{FL}, several families of codes are constructed using algebraic curves.

For an $[n,k,d]$ code, the well-known Singleton bound states that $d\leq n-k+1$. A code achieving the Singleton bound(i.e., $d=n-k+1$) is called maximum distance separable (MDS). Generalized Reed-Solomon codes are the most important realization of MDS codes. Recently, Euclidean self-dual MDS codes have received great attention. In \cite{JX}, a generic criterion is proposed to verify generalized Reed-Solomon (GRS) codes being self-dual. Later, it is generalized to the extended GRS case in \cite{HY}. In \cite{ZF}, a method to construct longer Euclidean MDS self-dual codes from shorter length is proposed, which is generalized in \cite{XFL}.  Additionally, several classes of MDS self-dual codes are constructed via multiplicative subgroups or additive subgroups in finite fields  (\cite{FLLL}, \cite{FZXF},   \cite{SYS}).  In particular, a large class of MDS self-dual codes are derived in \cite{WLZ} via the union of subsets from  two multiplicative subgroups in finite fields.

Self-orthogonal codes with respect to the Hermitian inner product play a crucial role in the construction of quantum codes (\cite{AK}, \cite{KKKS}). They have been extensively
studied in recent years (\cite{GLLS}). Notably, quantum MDS codes have been constructed using GRS codes, with the Hermitian construction proposed in
\cite{AK} in recent years, such as \cite{JLLX}, \cite{LLGS}, \cite{WT}.

Ashikhmin and Knill present the following Hermitian construction of quantum codes in \cite{AK}, which is useful for deriving quantum codes from classical codes.

\begin{theorem}{\bf (Hermitian Construction)}
    There is a $[[n, n - 2k, d^{\perp}]]_q$ quantum stabilizer code whenever there exists a classical Hermitian self-orthogonal $[n, k, d]_{q^2}$ linear code with dual distance $d^{\perp}$.
\end{theorem}

Using algebraic curves over finite fields, Goppa introduced an influential class of error-correcting codes known as algebraic-geometric (AG) codes \cite{G}. By leveraging the geometric properties of algebraic curves, AG codes achieve higher code rate while preserving larger minimum distance for the same code length. This allows them to exceed the Gilbert-Varshamov (GV) bound, a threshold that classical codes often struggle to reach. By effectively balancing code rate and minimum distance, AG codes deliver superior error-correction performance for a given code length. These properties can be applied to several communication scenarios such as optimal locally repairable codes \cite{MX}.

Based on the geometric properties of algebraic curves, such as selecting an appropriate set of rational points, we can construct a series of self-orthogonal algebraic geometry codes. For instance, certain Hermitian self-orthogonal codes with strong performance are derived from algebraic geometry codes \cite{JLLX}. In fact, by employing appropriate differential, we can characterize  algebraic geometry codes to be Euclidean or Hermitian self-orthogonal. Related work in recent years can be found in \cite{HMMF},\cite{J}. In particular, H. Stichtenoth showed in \cite{ST} that there exists a family of algebraic geometry codes achieving the TVZ bound that are equivalent to self-orthogonal codes.

In this paper, we construct several classes of self-orthogonal codes using two types of maximal curves. The parameters of these codes are either determined explicitly or estimated. The paper is organized as follows. In Section~$2$, we present some preliminary results on algebraic geometry (AG) codes that will be useful throughout the paper. In Section~$3$, we generalize the self-orthogonality criterion based on residues given in~\cite{S2}. In Section~$4$, we construct several families of Hermitian self-orthogonal AG codes from two types of maximal curves and provide some quantum codes with large minimum distances. In Section~$5$, we construct several families of Euclidean self-orthogonal AG codes from the same types of maximal curves, and we also present almost Euclidean self-dual codes with large minimum distances. In Section~$6$, we conclude the paper with a comparison of our results and previously known constructions.

\section{Preliminary}

\subsection{Algebraic functional fields}
\qquad In this section, we provide an introduction to algebraic functional fields. For more detailed information, see \cite{S}. Let $\mathbb{F}_q$ denote the finite field with $q$ elements, where $q$ is a power of a prime number $p$.

\begin{definition}

     An algebraic function field $F/\mathbb{F}_q$ of one variable over $\F_q$ is an extension field $F\supseteq \mathbb{F}_q$ such that $F$ is a finite algebraic extension of $\mathbb{F}_q(x)$ for some element $x\in F$ which is transcendental over $\mathbb{F}_q$. The algebraic closure of $\mathbb{F}_q$ in $F$ is called the constant field of $F/\mathbb{F}_q$. $\hfill\blacksquare$
\end{definition}

 Let $\chi$ be a smooth, projective, absolutely irreducible curve of genus $g$ defined over $\mathbb{F}_q$. Denote by $\mathbb{F}_q(\chi)$ the function field of $\mathbb{F}_q(\chi)/\mathbb{F}_q$. Since $\chi$ is absolutely irreducible, $\mathbb{F}_q$ is the constant field of $\mathbb{F}_q(\chi)$.

 \begin{definition}

    A place $P$ of the function field $\mathbb{F}_q(\chi)/\mathbb{F}_q$ is the maximal ideal of some valuation ring $\mathcal{O}$ which satisfies $\mathbb{F}_q\subsetneqq \mathcal{O}\subsetneqq \mathbb{F}_q(\chi)$. Every element $t \in P$ such that $P=t\mathcal{O}$ is called a uniformizing variable for $P$. Denote the set of all places of $\mathbb{F}_q(\chi)$ by $\mathbb{P} _{\F_q}(\chi)$. $\hfill\blacksquare$
\end{definition}

\begin{definition}

    Suppose that $P$ is a place of the function field $\mathbb{F}_q(\chi)/\mathbb{F}_q$, corresponding to the valuation ring $\mathcal{O}_P$. The degree of $P$ is then defined by
    $$
    \deg(P) = [\mathcal{O}_P/P : \mathbb{F}_q].
    $$
    Denote by $\chi(\mathbb{F}_q)$ the set of all places of degree one (also called rational places) of $\mathbb{F}_q(\chi)/\mathbb{F}_q$.
\end{definition}

Let $v_P$ denote the normalized discrete valuation associated with a point $P$ on the curve $\chi$. Throughout this paper, we adopt the simplified notation $(x-\alpha)$ to represent the place generated by $x-\alpha$ in the rational function field $\mathbb{F}_q(x)$, where $\alpha\in\mathbb{F}_q$. Similarly, let $P_\infty$ be the unique infinite place of $\mathbb{F}_q(x)$, corresponding to the simple pole of $x$.

Denote by $\#S$ the number of elements in a finite set $S$. Then $\chi$ is said to be maximal if it attains the upper Hasse-Weil bound: $\#\chi(\mathbb{F}_q)=q+1+2g\sqrt{q}$.

\begin{definition}

    A divisor of the function field $\mathbb{F}_q(\chi)/\mathbb{F}_q$ is a formal sum
    $$
    D = \sum_{P \in \mathbb{P} _{\F_q}(\chi)} n_P P, \quad n_P \in \mathbb{Z},
    $$
    where all but finitely many $n_P$ are zero.

    The support of $D$ is defined as
    $$
    \operatorname{supp}(D) := \{ P \in \mathbb{P} _{\F_q}(\chi) \mid n_P \neq 0 \}.
    $$
    For convenience, we may also write
    $$
    D = \sum_{P \in S} n_P P,
    $$
    where $S \subseteq \mathbb{P} _{\F_q}(\chi)$ is a finite set containing $\operatorname{supp}(D)$.
    The degree of $D$ is
    $$
    \deg(D) = \sum_{P \in \operatorname{supp}(D)} n_P\deg(P).
    $$
\end{definition}
Let $\Omega$ be the differential space of $\chi$. We define
$$
\Omega(G)=\{\omega\in\Omega\setminus\{0\}\mid\mathrm{div}(\omega)\geq G\}\cup \{0\},
$$
where $(\omega)$ is the canonical divisor corresponding to $\omega$. In fact, all canonical divisors have degree $2g-2$. Define the principal divisor of $x \in \mathbb{F}_q(\chi)$ as the
following:
$$
(x):=\sum\limits_{P \in \mathbb{P} _{\F_q}(\chi)} {{v_P}(x)P}.
$$
The Riemann-Roch space associated to $D$ is
$$
\mathscr{L}(D)=\{x\in \mathbb{F}_q(\chi)\mid(x)+D\geq 0\}\cup\{0\}.
$$
We denote the dimension of $\mathscr{L}(D)$ by $l(D)$.

The connection between the above concepts is elaborated by the following result.

 {\bf Riemann-Roch Theorem:}   Let $W$ be a canonical divisor of $\mathbb{F}_q(\chi)$. Then for each divisor $G$ of $\mathbb{F}_q(\chi)$, we have
    $$l(G)=\deg(G)+1-g+l(W-G).$$
Especially, we have $l(G)=\deg(G)+1-g$ if $\deg(G)> 2g-2$. $\hfill\blacksquare$

\subsection{GRS codes}
In this part, we will introduce some basic knowledge about generalized Reed-Solomon (GRS for short) codes.

For $1\leq n\leq q$, we choose $A=\{\alpha_1,\alpha_2,\ldots,\alpha_n\}$ and $V=(v_1,v_2,\ldots,v_n)$, where $\alpha_i\in\mathbb{F}_{q}$ are all distinct and $v_i\in\mathbb{F}_q^*$ for $1\leq i\leq n.$ Then the GRS code of length $n$ associated with $A$ and $V$ over $\mathbb{F}_q$ is
$${\rm GRS}_k(A,V,q)=\{(v_1f(\alpha_1),\ldots,v_nf(\alpha_n)):f(x)\in\mathbb{F}_q[x],\deg(f(x))\leq k-1\},$$
for $1\leq k\leq n$.

Moreover, the EGRS code (extended generalized Reed-Solomon codes) associated with $A$ and $V$ is defined by
$${\rm EGRS}_k(A,V,q)=\{(v_1f(\alpha_1),\ldots,v_nf(\alpha_n),f_{k-1}):f(x)\in\mathbb{F}_q[x],\deg(f(x))\leq k-1\},$$
where $1\leq k\leq n$ and $f_{k-1}$ is the coefficient of $x^{k-1}$ in $f(x)$.

\subsection{AG codes}

\qquad  Suppose that $P_{1},\ldots,P_{n}$ are pairwise distinct rational places of $\mathbb{F}_q(\chi)$. Suppose that $D=P_{1}+\dots+P_{n}$ and $G$ are divisors of $\mathbb{F}_q(\chi)/\mathbb{F}_q$ such that ${\rm supp}(G)\cap{\rm supp}(D)=\emptyset$.

Consider the evaluation map $\mathrm{ev}_{D}:\mathscr{L}(G)\to \mathbb{F}_q^{n}$ given by
$$ev_{D}(f)=(f(P_{1}),\ldots,f(P_{n}))\in \mathbb{F}_q^{n}$$
under the isomorphism between the residue fields of rational places and $\mathbb{F}_q$. The image of $\mathscr{L}(G)$ under $ev_{D}$ is the algebraic geometric code (or AG code
shortly) $C_\mathscr{L}(D,G)$. The Euclidean dual code of $C_\mathscr{L}(D,G)$ is
$$C_{\Omega}(D,G)=\{(\mathrm{res}_{P_1}(\omega),\ldots,\mathrm{res}_{P_n}(\omega))\mid\omega\in\Omega_{F}(G-D)\}.$$

\begin{theorem}\label{AG}{\rm(\cite{S} Theorems 2.2.2  and  2.2.7)}
    $C_{\mathscr{L}}(D,G)$ is an $[n,k,d]$ code with parameters
    \begin{equation*}\notag
        k=l(G)-l(G-D), d\geq n-\deg(G).
    \end{equation*}
   Its dual code $C_{\Omega}(D,G)$ is an $[n,k^{\prime},d^{\prime}]$ code with parameters
    \begin{equation*}
        k^{\prime}=n-l(G)+l(G-D), d^{\prime}\geq\deg(G)-(2g-2).
    \end{equation*}
    $\hfill\blacksquare$
\end{theorem}

The following result in \cite{S} shows the connection between $C_\mathscr{L}(D,G)$ and $C_\Omega(D,G)$.
\begin{theorem}\label{D}{\rm(\cite{S} Proposition 8.1.2)}
    Consider $D=\sum\limits_{i=1}^n P_i$, where each $P_i$ is a rational place. Let $\eta$ be a differential satisfying $v_{P_i}(\eta)=-1$ and $\operatorname{res}_{P_i}(\eta)=1$ for all $i=1,2,\ldots,n$. Then
    \begin{equation*}
    C_\Omega(D,G)=C_\mathscr{L}(D,D-G+(\eta)).
    \end{equation*}
     $\hfill\blacksquare$
\end{theorem}

In our construction, the following result from \cite{HY} is  instrumental.

\begin{lemma}\label{M}
Suppose that $m\mid q-1$ and $\alpha\in \mathbb{F}_q$ is an $m$-th primitive root. Then, for $1\leq i\leq m$:
\begin{equation*}
\prod\limits_{1\leq j\leq m,j \ne i} {({\alpha ^i} - {\alpha ^j})} = m{\alpha ^{ - i}}.
\end{equation*}
 $\hfill\blacksquare$
\end{lemma}

\section{Orthogonal Decision Criteria}
In this section, we present new criteria for constructing Hermitian and Euclidean self-orthogonal codes. Based on the connection
$C_\Omega(D,G)=C_\mathscr{L}(D,D-G+(\eta))$, we derive the following two key lemmas that play a central role in this paper.

Throughout, let $\chi$ denote a smooth, projective, absolutely irreducible curve of genus $g$ defined over $\mathbb{F}_q$, and let $\mathbb{F}_q(\chi)$ be its function field over $\mathbb{F}_q$. Let $D$ and $G$ be divisors of $\mathbb{F}_q(\chi)/\mathbb{F}_q$ such that ${\rm supp}(G) \cap {\rm supp}(D) = \emptyset$.

\begin{lemma}\label{ML}

For an AG code $C_\mathscr{L}$ over $\mathbb{F}_{q^2}$, assume that there is a differential $\eta$ with the following properties:
\begin{itemize}
    \item[{\rm(1).}] $(\eta)\geq (q+1)G-D$;
    \item[{\rm(2).}] For $1\leq i\leq n$, $\mathrm{res}_{P_i}\eta$ belongs to the same coset of $\mathbb{F}_q^*$.
\end{itemize}
Then there exists a divisor $G^{\prime}$ equivalent to $G$ such that $C_\mathscr{L}(D,G^{\prime})$ is Hermitian self-orthogonal.
\end{lemma}

\begin{proof}

    Denote by $\mathrm{res}_{P_i}\eta=hb_i^{q+1}\in h\mathbb{F}_q^*$ for some $b_i\in \mathbb{F}_{q^2}^*$.
    By the Weak Approximation Theorem, there exists $u\in \mathbb{F}_{q^2}(\chi)$ such that $v_{P_i}(u)=0$ and $u(P_i)=b_i$. Let $G^\prime=G-(u)$. Then, we obtain:
    $$(q+1)G^\prime-D=(q+1)G-D-(q+1)(u)\leq (u^{-(q+1)}\eta).$$
    Note that $f \in \mathscr{L}(G)$ implies  $f^q \in \mathscr{L}(qG)$. Since $\mathrm{res}_{P_i}(h^{-1}u^{-(q+1)}\eta)=1$ for $1\leq i\leq n$, we obtain:
    $$
    C_\mathscr{L}(D,G^{\prime})^q\subseteq C_\mathscr{L}(D,G^{\prime})^{\perp}=C_\mathscr{L}(D,D-G^{\prime}+(h^{-1}u^{-(q+1)}\eta))
    $$
    by Theorem \ref{D}.
    This establishes the conclusion.
\end{proof}

The above result can be viewed as a generalization of the following criterion from \cite{S2} to the Hermitian case:

\begin{lemma}

    {\rm(\cite{S2})} Assume that there is a differential $\eta$ with the following properties:
    \begin{itemize}
        \item[{\rm(1).}] $(\eta)\geq (q+1)G-D$;
        \item[{\rm(2).}] For $i=1,\cdots,n$, there is an element $b_i\in\mathbb{F}_q^*$ with ${\rm res}_{P_i}\eta=b_i^2$.
    \end{itemize}
    Then there exists a divisor $G^{\prime}$ equivalent to $G$ such that $C_\mathscr{L}(D,G^{\prime})$ is Euclidean self-orthogonal. Moreover, if we replace condition {\rm (1)} by  $(\eta)=2G-D$, then the code $C_\mathscr{L}(D, G^{\prime})$ is self-dual.    $\hfill\blacksquare$
\end{lemma}

Moreover, we can generalize the previous lemma in another way:

\begin{lemma}\label{ML2}
For an AG code $C_\mathscr{L}(D,G)$ over $\mathbb{F}_q$, if there is a differential $\eta$ with the following properties:
\begin{itemize}
    \item[{\rm(1).}] $(\eta)\geq 2G-D$;
    \item[{\rm(2).}] For $1\leq i\leq n$, $\mathrm{res}_{P_i}\eta$ are all squares or all non-squares.
\end{itemize}
then there exists a divisor $G^{\prime}$ equivalent to $G$ such that $C_\mathscr{L}(D,G^{\prime})$ is Euclidean self-orthogonal. Moreover, if we replace condition {\rm (1)} by  $(\eta)=2G-D$, then the code $C_\mathscr{L}(D, G^{\prime})$ is self-dual.
\end{lemma}

\begin{proof}

    The case where $\mathrm{res}_{P_i}\eta$ are all squares corresponds exactly to the previous lemma. Thus, we only need to consider the case where $\mathrm{res}_{P_i}\eta$ are all non-squares.

    Write $\mathrm{res}_{P_i}\eta = h b_i^2$, for some $b_i \in \mathbb{F}_q^*$ and $h$ is a primitive element. By the Weak Approximation Theorem, there exists $u \in \mathbb{F}_{q}(\chi)$ such that $v_{P_i}(u)=0$ and $u(P_i)=b_i$.
    Set $G' = G - (u)$. Then,
    $$
    2G' - D = 2G - D - 2(u) \leq (u^{-2}\eta).
    $$

    Since $\mathrm{res}_{P_i}(h^{-1}u^{-2}\eta) = 1$ for $1 \leq i \leq n$, we obtain
    $$
    C_\mathscr{L}(D,G') \subseteq C_\mathscr{L}(D,G')^{\perp} = C_\mathscr{L}(D,\, D - G' + (h^{-1}u^{-2}\eta)),
    $$
    by Theorem \ref{D}.
    This completes the proof.
\end{proof}

\section{Hermitian self-orthogonal codes}
In this section, we construct several Hermitian self-orthogonal codes and quantum codes by applying the new criterion introduced in Section 3. Throughout, for an irreducible curve $\chi$, we consistently treat the function field $\mathbb{F}_{q^2}(\chi)$ as an algebraic extension of $\mathbb{F}_{q^2}(x)$.

\subsection{Codes from $y^q+y=x^m$}

\qquad For $ m \mid q + 1 $, let $ \mathbb{F}_{q^2}(\chi) $ denote the function field of $ \chi $ over $ \mathbb{F}_{q^2} $, where $ \chi $ is defined by

$$\chi : y^q + y = x^m.$$

The genus of $ \chi $ is given by $ g = \frac{1}{2}(m-1)(q-1) $. A special case of this curve, where $ q $ is an odd power of 2 and $ m = 3 $, has been studied in \cite{J}, where it is used to construct quantum codes. Here we will investigate a general case.

Let $ n = q(m(q-1) + 1) $. Let $ \{ P_1, P_2, \dots, P_{m(q-1)} \} $ be the set of rational places of $ \mathbb{F}_{q^2}(x) $ corresponding to the $ m(q-1) $-th roots of unity in $ \mathbb{F}_{q^2} $. Let $ P_0 $ be the rational place $(x)$, and let $ P_\infty $ be the infinite place of $ \mathbb{F}_{q^2}(x) $.

In the extension $ \mathbb{F}_{q^2}(\chi)/\mathbb{F}_{q^2}(x) $, the finite places $ P_i $ split completely, and $ P_\infty $ is totally ramified. Let $ Q_{i,j} $ be the places of $ \mathbb{F}_{q^2}(\chi) $ that lie over $ P_i $ for $ 1 \leq i \leq m(q-1), 1\leq j\leq q$, and let $ Q_\infty $ be the infinite place lying over $ P_\infty $.

This means the number of rational places of $ \mathbb{F}_{q^2}(\chi)/\mathbb{F}_{q^2} $ is at least $ N = mq(q-1) + 1 = 2gq + q^2 + 1 $, which attains the upper Hasse-Weil bound. Therefore, $ \chi $ is a maximal curve.

Now, we can construct Hermitian self-orthogonal codes and quantum codes with good parameters.

\begin{theorem}\label{coro}

    Using the above notation, let $D=\sum\limits_{\scriptstyle1 \le i \le m(q - 1)\atop\scriptstyle1 \le j \le q} {{Q_{i,j}}}$. For $mq-m-q\leq r\leq m(q-1)-1$,
    \begin{itemize}
    \item[{\rm(1).}] There exists a divisor $G^{\prime}$ equivalent to $rQ_\infty$ such that $C_\mathscr{L}(D,G^{\prime})$ is Hermitian self-orthogonal with parameters $[mq^2-mq+q,k_0,d_0]_{q^2}$ with
    \[k_0=r-\frac{1}{2}(m-1)(q-1)+1, \quad d_0 \geq n-r.\]
    \item[{\rm(2).}] There exist  $q$-ary $[[mq^2-mq+q,k_1, d_1]]_q$ quantum codes with
      \[k_1=mq^2-m-2r-1,\quad d_1\geq r-mq+m+q+1.\]
    \end{itemize}
\end{theorem}

\begin{proof}

\begin{itemize}

  \item[(1).] Consider $\eta=\frac{-dx}{x(x^{m(q-1)}-1)}$. Then it is easy to see $v_{Q_{i,j}}(\eta)=-1$ for all $i,j$. According to Lemma \ref{M}, $\mathrm{res}_{Q_{i,j}}(\eta)=\frac{-1}{m(q-1)}$ for all $i,j$. They are both elements in $\F_q^*$. On the other hand,
  $$
  \begin{aligned}
  \nonumber
      (\eta) &= q(m(q-1)+1)Q_\infty-D+(dx)           \\
             &= (n+2g-2)Q_\infty-D.
  \end{aligned}
  $$
  The condition $r\leq m(q-1)-1$ implies $(q+1)r\leq n+2g-2$. Hence by Lemma \ref{ML}, there exists a divisor $G^{\prime}$ equivalent to $rQ_\infty$ such that $C_\mathscr{L}(D,G^{\prime})$ is Hermitian self-orthogonal with desired parameters.
  \item[(2).] Applying Hermitian construction to $C_\mathscr{L}(D,G^{\prime})$ in (1) produces the quantum code.
\end{itemize}
\end{proof}

\begin{example}

    Let $q=27, m=7$ and  $r=181$.  By Theorem \ref{coro} we obtain a $[4941,104,\geq 4760]_{729}$ code which is Hermitian self-orthogonal. Using Hermitian construction, we obtain quantum code with parameters $[[4941,4733,\geq 27]]_{27}$. $\hfill\blacksquare$
\end{example}

The following result for $q=2^l$ with $l$ odd and $m=3$ has been investigated in \cite{J}.
\begin{corollary}\label{coro_m}

    For odd $m$ and $q$ being a power of $2$, there exists a $q$-ary $[[mq^2-mq+q,k_1, d_1]]_q$ quantum code with $k_1=mq^2-m-2r-1,\quad d_1\geq r-mq+m+q+1$. Especially, let $m=3$ and $q$ is an odd power of $2$, we have quantum code with parameters $[[3q^2-2q,3q^2-4-2r,\geq r-2q+4]]_q$.
\end{corollary}

\begin{proof}

  For $q=p^l$ with $p=2$, $l$ odd,  note that $3\mid 2^l+1$ and $2\mid m-1$. Then the conclusion follows from  Theorem \ref{coro}.
\end{proof}

In comparison with the result of \cite{J}, the results in Theorem \ref{coro} extend the framework to more general algebraic curves, thereby removing the restriction that $q$ must be even. Consequently, Theorem \ref{coro} can be regarded as a natural generalization of Corollary \ref{coro_m}.

\begin{example}

For $q=8, m=3, r=20$, the corresponding code $C_\mathscr{L}(D,G^{\prime})$ is Hermitian self-orthogonal with parameters $[176,14,\geq 156]_{64}$. By Hermitian construction we obtain a quantum code with parameters $[[176,148,\geq 8]]_8$. $\hfill\blacksquare$
\end{example}

\subsection{Codes from Hermitian curves}

\qquad Denote by $\mathbb{F}_{q^2}(\chi)$ the Hermitian function field of $\chi$ over $\mathbb{F}_{q^2}$ with
$$\chi:y^q+y=x^{q+1}.$$
The genus of $\chi$ is given by $g = \frac{q(q-1)}{2}$.

For each $\alpha\in \mathbb{F}_{q^2}$, the places of $\mathbb{F}_{q^2}(\chi)$ lying over $(x-\alpha)$ split completely. Moreover, the simple pole $P_\infty$ of $x$ is totally ramified in $\mathbb{F}_{q^2}(\chi)$, and we denote the corresponding place in $\mathbb{F}_{q^2}(\chi)$ by $Q_\infty$. Hence $P_\infty=Q_\infty\cap \mathbb{F}_{q^2}(\chi)$.

\begin{theorem}\label{T5}

    There exists a Hermitian self-orthogonal code with parameters $[q^3-q^2,\frac{q(q-1)}{2},\geq q^3-2q^2+q-1]_{q^2}$.
\end{theorem}

\begin{proof}

    Denote by $h$ a primitive element in $\mathbb{F}_{q^2}$.
    Let $f(x)=\prod\limits_{\alpha\in\mathbb{F}_{q^2}\backslash \mathbb{F}_q} {( x - \alpha ) }$. Then, for $1\leq k\leq q$ and any $\alpha\in\mathbb{F}_q^*$, it follows that
    $$
    \left. { \frac{f(x)}{x-h^k\alpha}} \right|_{x=h^k\alpha}=\left. { \frac{x^{q^2}-x}{x-h^k\alpha}} \right|_{x=h^k\alpha}\cdot \left. { \frac{1}{\prod\limits_{\alpha\in\mathbb{F}_q}(x-\alpha)}} \right|_{x=h^k\alpha}=\left. { \frac{-1}{x^q-x}} \right|_{x=h^k\alpha}=\frac{-1}{(h^{kq}-h^k)\alpha}.
    $$
    In the above equation, the first equality holds since $1\leq k\leq q$ and any $\alpha\in\mathbb{F}_q^*$, we have $h^k\alpha\notin \mathbb{F}_q$.

    For $1\leq k_1\neq k_2\leq q$, observe that
    $$
    \left(\frac{h^{k_1q}-h^{k_1}}{h^{k_2q}-h^{k_2}}\right)^q=\frac{h^{k_1q^2}-h^{qk_1}}{h^{k_2q^2}-h^{qk_2}}=\frac{h^{k_1}-h^{qk_1}}{h^{k_2}-h^{qk_2}}=\frac{h^{k_1q}-h^{k_1}}{h^{k_2q}-h^{k_2}}.
    $$
    As a result, $\frac{1}{h^{k_1q}-h^{k_1}}$ and $\frac{1}{h^{k_2q}-h^{k_2}}$ lie in the same coset of $\mathbb{F}_q^*$.

    The places $(x-d)$ for $d\in \mathbb{F}_{q^2}\backslash \mathbb{F}_q$ split completely in $\mathbb{F}_{q^2}(\chi)$, and we denote the sum of these places in $\mathbb{F}_{q^2}(\chi)$ by $D$. Let $\eta=\frac{dx}{f(x)}$. Then, we obtain that ${\rm res}_P(\eta)$ lies in the same coset of $\mathbb{F}_q^*$ for all $P$ lying over $(x-d)$.

    Let $G=rQ_\infty$. For $r\leq q^2-q-1$, it holds that $(\eta)\geq (q+1)G-D$. Thus, by Lemma \ref{ML}, there exists a divisor $G^{\prime}$ equivalent to $G$ such that $C_\mathscr{L}(D,G^{\prime})$ is Hermitian self-orthogonal.

    Note that if $r= q^2-q-1$ , we obtain $r\geq 2g-1$, which implies that
    $${\rm dim}(C_\mathscr{L}(D,G^{\prime}))=l(rQ_\infty)=r+1-g.$$
    The second equality holds since $\deg(G^\prime)=\deg(rQ_\infty)\geq 2g-1$. This leads to the conclusion.
\end{proof}

\section{Euclidean self-orthogonal codes}
In this section, we construct several Euclidean self-orthogonal and Euclidean self-dual codes by applying the new criterion in Section 3. Throughout this section, for an irreducible curve $\chi$, we consistently treat $\mathbb{F}_{q^2}(\chi)/\mathbb{F}_{q^2}$ as an algebraic extension of $\mathbb{F}_{q^2}(x)/\mathbb{F}_{q^2}$.

As is well known, Euclidean self-dual codes must have even length. Therefore, a code of odd length cannot be Euclidean self-dual. However, there exists a class of codes called almost self-dual codes, which have odd length:

\begin{definition}
A code with parameters $[2k+1, k, d]_q$ is called an almost self-dual code if it is Euclidean self-orthogonal.
\end{definition}

In fact, such codes achieve the maximal possible dimension under the condition of Euclidean self-orthogonality.

\subsection{From GRS codes and EGRS codes}

Recall the definitions of GRS codes and EGRS codes:

For $1\leq n\leq q$, we choose $A=\{\alpha_1,\alpha_2,\ldots,\alpha_n\}$ and $V=(v_1,v_2,\ldots,v_n)$, where $\alpha_i\in\mathbb{F}_{q}$ are all distinct and $v_i\in\mathbb{F}_q^*$ for $1\leq i\leq n.$ Then
$$
{\rm GRS}_k(A,V,q)=\{(v_1f(\alpha_1),\ldots,v_nf(\alpha_n)):f(x)\in\mathbb{F}_q[x],\deg(f(x))\leq k-1\},
$$
for $1\leq k\leq n$.
$$
{\rm EGRS}_k(A,V,q)=\{(v_1f(\alpha_1),\ldots,v_nf(\alpha_n),f_{k-1}):f(x)\in\mathbb{F}_q[x],\deg(f(x))\leq k-1\},
$$
where $1\leq k\leq n$ and $f_{k-1}$ is the coefficient of $x^{k-1}$ in $f(x)$.

For ${\rm GRS}_{\frac{n}{2}}(A, V)$ (or ${\rm EGRS}_{\frac{n+1}{2}}(A, V, q)$), let $A=\{\alpha_1,\alpha_2,\ldots,\alpha_n\}$ and $V=(v_1,v_2,\ldots,v_n)$, and define
$$\Delta_A(\alpha_i)=\prod_{1\leq j\leq n,j\neq i}(\alpha_i-\alpha_j).$$
Then the following result provides a criterion for determining whether a GRS code or an EGRS code is self-orthogonal.

\begin{lemma}\label{DG}

    {\rm(\cite{JX})} Suppose that $n$ is even. There exists $V=(v_1, v_2, \cdots, v_n)\in(\mathbb{F}_q^*)^n$ such that ${\rm GRS}_{\frac n2}(A, V)$ is self-dual if and only if all $\Delta_A(\alpha)$ are square for any $\alpha\in A$. $\hfill\blacksquare$
\end{lemma}

\begin{lemma}\label{DEG}

    {\rm(\cite{HY})} Suppose that $n$ is even. There exists $V=(v_1, v_2, \cdots,v_n)\in(\mathbb{F}_q^*)^n$ such that ${\rm EGRS}_{\frac{n+1}{2}}(A, V, q)$ is self-dual if and only if all $-\Delta_A(\alpha)$ are square for any $\alpha\in A$. $\hfill\blacksquare$
\end{lemma}

In \cite{XFL}, several methods for constructing long MDS self-dual codes were introduced. While all of these codes are MDS, their lengths are inherently limited due to the constraints of the MDS property. Here we generalize these methods to AG codes, enabling the construction of self-dual codes with significantly larger lengths. Moreover, even for some odd lengths $n$, we can construct almost self-dual AG codes with parameters $[n, \frac{n-1}{2}]$.

The following notation will be used in this subsection (Subsection 5.1):

\begin{itemize}
    \item $q = p^m$.
    \item $H$ is an $\mathbb{F}_{p^s}$-subspace of $\mathbb{F}_q$ with dimension $l$, where $s \mid 2m$.
    \item $A = \{a_1, a_2, \dots, a_{2t}\}$ and $B = \bigcup\limits_{a_i \in A}(a_i \beta + H) = \{b_1, b_2, \dots, b_{n_0}\}$, where $2t \leq p^s$ and $\beta\in\F_q\backslash H$.
    \item $n_0 = (2t + 1)p^{sl}$.
\end{itemize}

We are now ready to construct self-dual (or almost self-dual) AG codes from self-dual EGRS/GRS codes.

\begin{theorem}\label{T6}

    Suppose that there exists a smooth, projective, absolutely irreducible curve $\chi$ of genus $g$ such that the places $(x - b_i)$ split completely and $P_\infty$ is totally ramified over $\mathbb{F}_{q^2}(x)$ (for example, the Hermitian curve). Let $[\mathbb{F}_{q^2}(\chi):\mathbb{F}_{q^2}(x)] = \lambda$. Assume that
    $$
    2g - 1 \leq r \leq \left\lfloor \frac{n + d(Q_\infty \mid P_\infty)}{2} \right\rfloor,
    $$
    where $d(Q_\infty \mid P_\infty)$ denotes the different exponent of $Q_\infty$ over $P_\infty$ (this inequality is satisfied in the case of Hermitian curves), and that there exists a self-dual EGRS code with parameters $[2t+2,\ t+1,\ t+2]_{q^2}$. Then there exist (almost) self-dual codes with parameters $[n, k, \geq n-k-g+1]_{q^2}$, where
    $$
    n = \lambda n_0, \qquad k = \left\lfloor \tfrac{n}{2} \right\rfloor.
    $$
    In particular, these codes are self-dual when $q$ is odd and almost self-dual when $q$ is even.
\end{theorem}

\begin{proof}

    We choose $V \in (\mathbb{F}_{q^2}^*)^{2t+1}$ such that ${\rm EGRS}(A, V, q^2)$ is a self-dual code.

    According to Lemma \ref{DEG}, if $-\Delta_A(\alpha)$ are all squares for $\alpha\in A$, then we obtain
    $$\Delta_B(b_i)=-\Delta_H(0)\Delta_A(a_k)(f_H(\alpha))^{2t},$$
    where $f_H(x)=\prod\limits_{h \in H} {(x - h)}$. Consequently, $\Delta_B(b_i)$ are all squares.

    Next, consider the places of $\mathbb{F}_{q^2}(\chi)$ lying over $(x-b)$ with $b \in B$: they are all completely split. Let $D$ be the sum of these places. Now, for AG codes $C_{\mathscr{L}}(D,rQ_\infty)$, define $h(x)=\prod\limits_{b \in B} {(x - b)}$ and $\eta=\frac{dx}{h(x)}$. The residues ${\rm res}_{P_i}\eta$ are all squares, and it follows that
    $$(\eta)=-D+nQ_\infty + (dx).$$
    Observe that the complete splitting of the places $(x - b_i)$, the total ramification of $P_\infty$, and the condition $r \leq \left\lfloor \tfrac{n + d(Q_\infty \mid P_\infty)}{2} \right\rfloor$ imply that $(\eta) \geq 2G - D$. Therefore, by utilizing Lemma \ref{ML2}, there exists a divisor $G^{\prime}$ equivalent to $rQ_\infty$ such that $C_\mathscr{L}(D,G^{\prime})$ is Euclidean self-orthogonal. Moreover, the condition $\lambda n_0\geq 2g-1$ ensures that the dimension of $C_\mathscr{L}(D,G^{\prime})$ is equal to $\left\lfloor\frac{n}{2}\right\rfloor$ by the Riemann-Roch Theorem. Moreover, if $q$ is odd, the equality $2rQ_\infty = D + (\eta)$ holds, which, by Lemma \ref{ML2}, implies the existence of Euclidean self-dual codes with parameters
    $$[\lambda(2t+1)p^{sl},\frac{\lambda(2t+1)p^{sl}}{2},\geq \lambda(2t+1)p^{sl}-r]_{q^2}.$$
\end{proof}

\begin{remark}
The above result provides a generic method to construct (almost) self-dual AG codes from self-dual GRS codes. By this way, many classes of  (almost) self-dual AG codes can be obtained from known self-dual GRS codes.
\end{remark}

Similarly, the following can be shown
\begin{theorem}\label{T7}

    Suppose that there exists a smooth, projective, absolutely irreducible curve $\chi$ of genus $g$ such that the places $(x - b_i)$ split completely and $P_\infty$ is totally ramified over $\mathbb{F}_{q^2}(x)$ (for example, the Hermitian curve). Let $[\mathbb{F}_{q^2}(\chi):\mathbb{F}_{q^2}(x)] = \lambda$. Assume that
    $$
    2g - 1 \leq r \leq \left\lfloor \frac{2\lambda tp^{sl} + d(Q_\infty \mid P_\infty)}{2} \right\rfloor,
    $$
    where $d(Q_\infty \mid P_\infty)$ denotes the different exponent of $Q_\infty$ over $P_\infty$ (this inequality is satisfied in the case of Hermitian curves), and that there exists a self-dual ${\rm GRS}$ code of length $2tp^{sl}$, then we obtain Euclidean self-orthogonal codes with parameters
    $$[2\lambda tp^{sl},r+1-g, \geq 2\lambda tp^{sl}-r]_{q^2}.$$
    for $2g-1\leq r\leq g-1+\lambda tp^{sl}$. In particular, there exist Euclidean self-dual codes with parameters $[2\lambda tp^{sl},\lambda tp^{sl},\geq \lambda tp^{sl}+1-g]_{q^2}$. $\hfill\blacksquare$
\end{theorem}

\begin{proof}

    Similar to the previous proof, by Lemma \ref{DG}, if $\Delta_A(\alpha)$ is a square for every $\alpha \in A$, then
    $$
    \Delta_B(b_i) = -\Delta_H(0)\,\Delta_A(a_k)\,(f_H(\alpha))^{2t},
    $$
    where
    $$
    f_H(x) = \prod_{h \in H} (x - h).
    $$
    Hence, $\Delta_B(b_i)$ are either all squares or all non-squares.

    Now consider the places of $\mathbb{F}_{q^2}(\chi)$ lying over $(x - b_i)$. Let $D$ be the sum of these places, and set $G = rQ_\infty$. The condition $r \leq \left\lfloor \frac{2\lambda t p^{sl} + d(Q_\infty \mid P_\infty)}{2} \right\rfloor$ implies that $(\eta) \geq 2G - D$. By Lemma~\ref{ML2}, there then exists a divisor $G^\prime$ linearly equivalent to $G$ such that $C_\mathscr{L}(D, G^\prime)$ is Euclidean self-orthogonal. In particular, self-dual codes exist in this construction.
\end{proof}

\begin{example}
    If self-dual EGRS codes with parameters $[2t+2, t+1, t+2]_{q^2}$ exist, then several families of Euclidean self-dual and almost self-dual codes can be constructed from the Hermitian curve, whose genus is $g = \tfrac{q(q-1)}{2}$, as follows:
    \begin{itemize}
        \item [\rm(1).] There exist Euclidean self-orthogonal codes with parameters
        $$\left[q(2t+1)p^{sl},r+1-\frac{q(q-1)}{2},\geq q(2t+1)p^{sl}-r\right]_{q^2}$$
        for $q(q-1)-1\leq r\leq \frac{q(q-1)}{2}-1+\lfloor \frac{q(2t+1)p^{sl}}{2}\rfloor $.
        \item [\rm(2).] There exist Euclidean self-dual codes with parameters
        $$\left[q(2t+1)p^{sl},\frac{q(2t+1)p^{sl}}{2},\geq \frac{q(2t+1)p^{sl}}{2}-\frac{q(q-1)}{2}+1\right]_{q^2}$$
        for even $q$.
        \item [\rm(3).] In particular, there exist Euclidean self-orthogonal codes with parameters
        $$\left[q(2t+1)p^{sl},\frac{q(2t+1)p^{sl}-1}{2},\geq \frac{q(2t+1)p^{sl}+1}{2}+1-\frac{q(q-1)}{2}\right]_{q^2}$$
        for odd $q$. This is also a Euclidean self-orthogonal code with the largest possible dimension for codes of odd length.
    \end{itemize}
     $\hfill\blacksquare$
\end{example}

In contrast to \cite{XFL}, where MDS self-dual codes of length $(2t+1)p^{sl}+1$ were obtained, the self-dual and almost self-dual codes constructed here do not retain the MDS property. However, their length is $q$ times that of the codes in \cite{XFL}. Furthermore, \cite{Sok} proved the existence of self-dual codes of length $q(2t+1)p^{sl}$ for even $q$. We generalize this result to odd $q$ by constructing almost self-dual codes.

\subsection{Codes from $y^q+y=x^m$}

\qquad   For $m\mid q+1$, denote by $\mathbb{F}_{q^2}(\chi)$ the function field of $\chi$ over $\mathbb{F}_{q^2}$, where
$$\chi:y^q+y=x^m.$$
Recall the notation and conclusion in Section 4.1. This corresponds to a maximal curve with $g=\frac{1}{2}(m-1)(q-1)$.
Let $n=q(m(q-1)+1)$. Let $\{P_0,P_1,P_2,\cdots,P_{m(q-1)}\}$ be the set of all places of $\mathbb{F}_{q^2}(x)$ that split completely in $\mathbb{F}_{q^2}(\chi)$, where $P_0$ is the rational place $(x)$. Let $P_\infty$ be the infinite place of $\mathbb{F}_{q^2}(x)$, and let $Q_{\infty}$ be the infinite place lying over $P_{\infty}$. Let $D=\sum\limits_{i=0}^{m(q-1)}\sum\limits_{j=1}^q {{Q_{i,j}}}$ be the sum of all places lying over $\{P_0,P_1,P_2,\cdots,P_{m(q-1)}\}$.

Now we can construct Euclidean self-orthogonal and self-dual codes with good parameters.
\begin{theorem}\label{mc}

    \begin{itemize}
        \item[{\rm(1).}]For $mq-m-q\leq r\leq\lfloor\frac{m(q^2-1)-1}{2}\rfloor$, there is a Euclidean self-orthogonal code with parameters $[n,k_0,d_0]_{q^2}$ where
        $$n=mq^2-mq+q,\quad k_0\geq r-\frac{1}{2}(m-1)(q-1)+1, \quad d_0\geq n-r.$$
        \item[{\rm(2).}]In particular, there exists a Euclidean self-dual code with parameters
        $$\left[mq^2-mq+q,\frac{mq^2-mq+q}{2},\geq \frac{mq^2-mq+q}{2}-\frac{(m-1)(q-1)}{2}+1\right]_{q^2}$$
        when $q$ is even.
    \end{itemize}
\end{theorem}

\begin{proof}

\begin{itemize}
  \item[(1).] Consider $\eta=\frac{-dx}{x(x^{m(q-1)}-1)}$. Then it is easy to see $v_{Q_{i,j}}(\eta)=-1$ for all $i,j$. According to Lemma \ref{M}, ${\rm res}_{Q_{i,j}}(\eta)=\frac{-1}{m(q-1)}$ for all $i,j$. Since these residue values are identical, it follows that they are either all square elements or all non-square elements. On the other hand,
  $$
  \begin{aligned}
  \nonumber
      (\eta) &= q(m(q-1)+1)Q_\infty-D+(dx)           \\
             &= (n+2g-2)Q_\infty-D.
  \end{aligned}
  $$
  Since $2r\leq n+2g-2=m(q^2-1)-1$, the conditions in Lemma \ref{ML2} are satisfied. In this case, there exists a divisor $G^{\prime}$ equivalent to $rQ_\infty$ such that the Euclidean self-orthogonal code $C_\mathscr{L}(D,G^{\prime})$ has length $n$. Since $2g-2<r<\deg D=mq^2-mq+q$ and $G^{\prime}$ is equivalent to $rQ_\infty$, by Riemann-Roch Theorem the dimension of $C_\mathscr{L}(D,G^{\prime})$ is
  $$
  k_0=l(rQ_{\infty})-l(rQ_{\infty}-D)\geq r-\frac{1}{2}(m-1)(q-1)+1.
  $$
  Since $r < \deg D$ implies $l(rQ_{\infty}-D) = 0$. Moreover, for $r>2g-2=mq-m-q-1$, the equality holds, that is, $k_0=r-\frac{1}{2}(m-1)(q-1)+1$. Its minimal distance is at least $n-\deg(rQ_{\infty})=n-r$.
  \item[(2).] Consider $2r={m(q^2-1)-1}$. Then there exists a divisor $G^{\prime}$ equivalent to $rQ_\infty$ such that $C_\mathscr{L}(D,G^{\prime})=C_\mathscr{L}(D,G^{\prime})^{\perp}$. Then $C_\mathscr{L}(D,G^{\prime})$ is self-dual with desired parameters.
\end{itemize}
\end{proof}

In this way, we obtain a $\left[mq^2-mq+q,\frac{mq^2-mq+q}{2},\geq \frac{mq^2-mq+q}{2}-\frac{(m-1)(q-1)}{2}+1\right]_{q^2}$ self-dual code with minimal distance close to $n/2$, since $\frac{(m-1)(q-1)}{2}+1$ is relatively small compared to $n$ when $q$ is large.

\begin{example}\label{TE}

    Note that $\frac{n-1}{2}=\frac{(mq+1)(q-1)}{2}$. Let $r_0=\frac{mq^2-m}{2}-1$, yielding $r_0\geq 2g-1$. Consequently, by the Riemann-Roch Theorem,
    $$l(r_0Q_\infty)=r_0+1-\frac{(m-1)(q-1)}{2}=\frac{mq^2-m}{2}-\frac{(m-1)(q-1)}{2}=\frac{mq^2-mq+q-1}{2}.$$
    This implies that $C_\mathscr{L}(D,G^{\prime})$ is a Euclidean self-orthogonal code with parameters $[n,\frac{n-1}{2},\geq \frac{1}{2}(mq^2+m)-(m-1)q+1]_{q^2}$. For example, let $q=5, m=3$. Then we obtain a Euclidean almost self-dual code with parameters $[n,k,d]_{q^2}=[65,32,\geq 30]_{25}$, which satisfy $k+d\geq 62$, close to the code length $n=65$.
    $\hfill\blacksquare$
\end{example}

\subsection{Codes from Hermitian curves}

\qquad Denote by $\mathbb{F}_{q^2}(\chi)$ the Hermitian function field of $\chi$ over $\mathbb{F}_{q^2}$ with
$$\chi:y^q+y=x^{q+1}.$$
The genus of $\chi$ is given by $g = \frac{q(q-1)}{2}$. The simple pole $P_\infty$ of $x$ in $\mathbb{F}_{q^2}(x)$ is totally ramified in $\mathbb{F}_{q^2}(\chi)$, and the place lying over $P_\infty$ is denoted by $Q_\infty$. The places $(x-\alpha)$ split completely for any $\alpha \in \mathbb{F}_{q^2}$.

\subsubsection{From multiplicative groups}

\qquad For $s\mid q^2-1$, where $s$ is a square element, let $n=q(s+1)$. Let $\{P_1,P_2,\cdots,P_s\}$ be the set of rational places of $\mathbb{F}_{q^2}(x)$ corresponding to the $s$-th roots of unity in $\mathbb{F}_{q^2}$. Let $P_0$ be the rational place $(x)$.

Then these finite places $P_i$ split completely. Let $Q_{i ,j}$ be all places of $\mathbb{F}_{q^2}(\chi)$ that lie over $P_i$. Let $$D=\sum\limits_{i=0}^{s}\sum\limits_{j=1}^q{{Q_{i,j}}}.$$

\begin{theorem}\label{q(q+1)}

    \begin{itemize}
        \item[{\rm(1).}] There exist Euclidean self-orthogonal codes with parameters
        $$\left[q(s+1),r-\frac{1}{2}q(q-1)+1,\geq q(s+1)-r\right]_{q^2}$$
        for $q^2-q-1\leq r\leq\lfloor\frac{1}{2}q(q+s)-1\rfloor$.
        \item[{\rm(2).}] There exists a Euclidean self-dual code with parameters
        $$\left[q(s+1),\frac{q(s+1)}{2}, \geq \frac{q(s+1)}{2}-\frac{q(q-1)}{2}+1\right]_{q^2}$$
        when $q$ is even or $s$ is odd.
    \end{itemize}
\end{theorem}

\begin{proof}

    \begin{itemize}
        \item[(1).] The proof follows similarly to that of Theorem \ref{mc}. Denote by $n=q(s+1)$.
        Consider $\eta=\frac{-dx}{x(x^s-1)}$. It is easy to see $v_{Q_{i,j}}(\eta)=-1$. According to Lemma \ref{M},
        $$
        \mathrm{res}_{Q_{i,j}}(\eta)=
        \left\{
            \begin{aligned}
            \nonumber
                &-\frac{1}{s}        &   i \neq 0,            \\
                &1                     &   i=0.                 \\
            \end{aligned}
        \right.
        $$
        If $q$ is odd, all elements in $\F_{q}$ are square elements in $\F_{q^2}$. Consequently, the residues $\mathrm{res}_{Q_i}(\eta)$ are also square elements in $\F_{q^2}$.

        On the other hand, let $G=rQ_\infty$. It follows that
        $$
        \begin{aligned}
        \nonumber
            (\eta) &= q(s+1)Q_\infty-D+(dx)           \\
                   &= (n+2g-2)Q_\infty-D,
        \end{aligned}
        $$
        which implies $D-G+(\eta)=(n+2g-2-r)Q_\infty\geq G$ since $2r\leq n+2g-2$. Therefore $(\eta)\geq 2G-D$ and  there exists a divisor $G^{\prime}$ equivalent to $rQ_\infty$ such that $C_\mathscr{L}(D,G^{\prime})$ is Euclidean self-orthogonal according to Lemma \ref{ML2}. In this case, the code  $C_\mathscr{L}(D,G^{\prime})$ has length $n$, minimal distance at least $n-\deg(rQ_\infty)=n-r$. The dimension $k_0$ can be derived from Theorem \ref{AG} and Riemann-Roch Theorem.
        \item[(2). ] Consider $2r={q(q+s)-2}$. Then $C_\mathscr{L}(D,G^{\prime})$ is Euclidean self-dual by utilizing Lemma \ref{ML2}. Then $C_\mathscr{L}(D,G^{\prime})$ is self-dual with desired parameters.
    \end{itemize}
\end{proof}

\begin{example}
    Given that $q$ is odd and $s$ is even, it follows that $q(s+1)$ is odd. Moreover, when $s\geq q-2$, we can establish the existence of Euclidean self-orthogonal codes $\left[q(s+1),\frac{q(s+1)-1}{2}, \geq \frac{1}{2}q(s-q+2)+\frac{3}{2}\right]_{q^2}$ for similar reasons as demonstrated in Example \ref{TE}. For example, let $q=25, s=338$. Then we obtain a Euclidean almost self-dual code with parameters $[n,k,d]_{q^2}=[8475,4237,\geq 3939]_{625}$, which satisfy $k+d\geq 8176$, close to the code length $n=8475$.
    $\hfill\blacksquare$
\end{example}

\subsubsection{From additive groups}

\begin{theorem}\label{T10}

    For $(q+1)\mid k$, it follows that
    \begin{itemize}
        \item[\rm(1).] There exist Euclidean self-orthogonal codes with parameters
        $$\left[q^3-qk,r-\frac{q(q-1)}{2}+1,\geq q^3+\frac{q(q-1)}{2}-qk-1\right]_{q^2}$$
        for $q^2-q-1\leq r\leq \lfloor \frac{q^3-qk}{2}\rfloor+\frac{q(q-1)}{2}-1$.
        \item[\rm(2).] There exists a Euclidean self-dual code with parameters
        $$\left[q^3-qk,\frac{q^3-qk}{2}, \geq \frac{q^3-qk}{2}-\frac{q(q-1)}{2}+1\right]_{q^2}$$
        when $q$ is even or $k$ is odd.
    \end{itemize}
\end{theorem}

\begin{proof}

    \begin{itemize}
        \item [(1).] Denote by $U_k$ the set of $k$-th roots of unity in $\F_{q^2}^*$.
        The places $(x-d)$ for $d\in \mathbb{F}_{q^2}\backslash U_k$ split completely in $\mathbb{F}_{q^2}(\chi)$, and we denote the sum of those places of $\mathbb{F}_{q^2}(\chi)$ by $D$. Let $\eta=\frac{dx}{f(x)}=\frac{dx}{\prod\limits_{\alpha\in\mathbb{F}_{q^2}\backslash U_k} {( x - \alpha ) }}$. For $\alpha\in\mathbb{F}_{q^2}\backslash U_k$, we have
        $$
        \left. {\frac{f(x)}{x-\alpha}} \right|_{x=\alpha}=\left. {\frac{x^{q^2}-x}{x-\alpha}} \right|_{x=\alpha}\cdot \left. {\frac{1}{\prod\limits_{\alpha\in U_k}(x-\alpha)}} \right|_{x=\alpha}= \left. {\frac{-1}{x^k-1}} \right|_{x=\alpha}=\frac{-1}{\alpha^k-1}.
        $$

        On the other hand, if $r \leq g - 1 + \left\lfloor \frac{q^3 - qk}{2} \right\rfloor$, this implies $G \leq D - G + (\eta)$. Therefore, there exists a divisor $G'$ equivalent to $rQ_\infty$ such that $C_\mathscr{L}(D, G')$ is Euclidean self-orthogonal. Here, the condition $2g - 1 \leq r$ gives $l(G^\prime)=l(rQ_\infty) = r + 1 - g = r - \frac{1}{2} q(q - 1) + 1$ by the Riemann-Roch Theorem. Thus, the conclusion follows.
        \item[(2).] Consider $2r = q^3 - qk + q(q - 1) - 2$. Under this setting, the stronger condition in Lemma~\ref{ML2} holds, and consequently, $C_\mathscr{L}(D, G')$ is self-dual with the desired parameters.
    \end{itemize}
\end{proof}

\begin{example}
    Given that $q$ is odd, $k$ is even, and $(q+1)\mid k$, it follows that $q^3-qk$ is odd. If $q^3-qk\geq q^2-q-1$, we can construct  Euclidean self-orthogonal codes $[q^3-qk,\frac{q^3-qk-1}{2}, \geq \frac{q^3-qk}{2}-\frac{q(q-1)}{2}+1]_{q^2}$ for similar reasons as demonstrated in Example \ref{TE}. For example, let $q=9, k=70$. Then we obtain a Euclidean almost self-dual code with parameters $[99,49,\geq 32]_{81}$.
    $\hfill\blacksquare$
\end{example}

\qquad Suppose that $q=p^t$. For $k\leq2t$, let $V_k$ be a $k$-dimensional $\mathbb{F}_p$-subspace of $\mathbb{F}_{q^2}$.
Let $n=q\cdot p^k$ and let $\{P_1,P_2,\cdots,P_{p^k}\}$ be the set of rational places corresponding to the elements in $V_k$. Let $P_0$ be the rational place $(x)$.

Similarly to the previous case, all finite places $P_i$ split completely in $\mathbb{F}_{q^2}(\chi)$. Let $Q_{i ,j}$ be all places of $\mathbb{F}_{q^2}(\chi)$ that lie over $P_i$. Denote by $D=\sum\limits_{i = 1}^{p^k}\sum\limits_{j=1}^q {{Q_{i,j}}}$.

\begin{theorem}\label{T11}

    \begin{itemize}
        \item[{\rm(1).}] There exist Euclidean self-orthogonal codes with parameters
        $$\left[p^kq,r-\frac{1}{2}q(q-1)+1,\geq p^kq-r\right]_{q^2}$$
        when $q^2-q-1\leq r\leq\lfloor\frac{q(p^k+q-1)}{2}-1\rfloor$.
        \item[{\rm(2).}] There exists a Euclidean self-dual code with parameters
        $$\left[p^kq,\frac{p^kq}{2},\geq \frac{ p^kq}{2}-\frac{q(q-1)}{2}+1\right]_{q^2}$$
        when $q$ is even.
    \end{itemize}
\end{theorem}

\begin{proof}

    \begin{itemize}
        \item[(1).] Let $n = p^k q$. Define $h(x) = \prod\limits_{\alpha \in V_k} (x - \alpha)$. For each $\alpha \in V_k$, consider $h_\alpha(x) = \frac{h(x)}{x - \alpha}$. Since $V_k$ is stable under translation by elements of $V_k$, the value $h_\alpha(x)\big|_{x=\alpha}$ is constant for all $Q_{i,j}$ lying over $P_i$ with $1 \leq i \leq n$. Denote this constant by $\lambda$. Then $\lambda = \operatorname{res}_{Q_{i,j}}(h(x))$. Let $\eta = \frac{\lambda dx}{h(x)}$. Then $(\eta)=(n+2g-2)Q_{\infty}-D$. Since $2r \leq n + 2g - 2$, there exists a divisor $G'$ equivalent to $r Q_\infty$ such that $C_\mathscr{L}(D, G')$ is Euclidean self-orthogonal by Lemma~\ref{ML2}.

        In this case, the code $C_\mathscr{L}(D, G')$ has length $n$, minimal distance at least $n - \deg(r Q_\infty) = n - r$.
        The dimension $k_0$ can be derived from Theorem~\ref{AG} and the Riemann-Roch Theorem.
        \item[(2).] Consider $2r = p^k q + q(q - 1) - 2$. Then the stronger condition in Lemma~\ref{ML2} holds. Consequently, $C_\mathscr{L}(D, G')$ is self-dual with the desired parameters.
    \end{itemize}
\end{proof}

\begin{example}
    Given that $q$ is odd and $p^k\geq q-1$, it follows that $p^kq$ is odd. Moreover, with $p^k\geq q-1$, we can establish the existence of Euclidean self-orthogonal codes $\left[p^kq,\frac{p^kq-1}{2},\geq \frac{p^kq-q(q-1)+3}{2}\right]_{q^2}$ for similar reasons as demonstrated in Example \ref{TE}. For example, let $p=3, q=27, k=5$. Then we obtain a Euclidean almost self-dual code with parameters $[n,k_0,d]_{q^2}=[6561,3280,\geq 2931]_{729}$, which satisfies $k_0+d\geq 6211$, close to the code length $n=6561$.
    $\hfill\blacksquare$
\end{example}

\section{Conclusion and further study}

\qquad In this paper, we generalized the criterion for self-orthogonality and applied it to construct self-orthogonal and self-dual codes over function fields arising from Hermitian curves and their quotients over $\mathbb{F}_{q^2}$. These codes demonstrated both strong parameters and desirable Euclidean or Hermitian self-orthogonal (or self-dual) properties. The following tables present a comparison of our results with existing findings.

In Table~1, the quantum codes
$[[mq^2 - mq + q,\, mq^2 - m - 2r - 1,\, \geq r - mq + m + q + 1]]_q$
generalize the results in \cite{J}, corresponding to the special case where $q = 2^\ell$ with $\ell$ odd and $m = 3$. 
In contrast, our construction yields $q - 1$ quantum codes for each $q$ (both even and odd) with large lengths, which has received limited attention in the literature.

In Table 2, our construction yields self-dual codes with greater parameter flexibility and longer lengths for even $q$ compared to the results in \cite{Sok}. Additionally, relative to the findings in \cite{FL}, we impose fewer restrictions on the parameters and provide different code lengths.

Furthermore, for odd $q$, Euclidean self-dual codes do not exist in some cases because $n$ is odd. Nevertheless, we construct Euclidean almost self-dual codes with improved parameters for even lengths and even $q$, an area with limited related studies, as shown in Table 3.

\begin{sidewaystable}[thp]
    \renewcommand{\arraystretch}{2.0}
    \caption{Some known results on quantum codes from algebraic curves.}
    \centering
    \begin{tabular}{c|c|c|c} 
        \textbf{$q$} & {$Code$} & \textbf{algebraic curves} & \textbf{Reference}\\
        \hline
        $2\mid q$, $q-1\leq m\leq 2q-2$  &  $[[2q^2,2q^2-2m+q-2,\geq m+2-q]]_q$   &   $y^2+y=x^{q+1}$     &   \cite{J}  \\
        $2\mid q$, $2q-3\leq m\leq 3q-4$  &  $[[3q^2-2q,3q^2-2m-4,\geq m+4-2q]]_q$   &   $y^q+y=x^3$     &   \cite{J}  \\
        $2m+2\leq n\leq q+\lfloor 2\sqrt{q} \rfloor -5$  &  $[[n,n-2m,m]]_q$   &   elliptic curve     &   \cite{JX2} \\
        $\sim$  &  \makecell[c]{$[[q^8-q^6+q^5,q^8-q^6+$\\$2q^5-2q^3+q^2-2-2s,\geq s-q^5+2q^3-q^2+2]]_{q^6}$}  &    GK curve    &   \cite{BMZ}  \\
        \ &\ &\ &\ \\
        \makecell[c]{$(q-1)(q^{n+1}+q^n-q^2)-2\leq s\leq$\\$\left((q^{n-1}-1)\frac{q^n+1}{q+1}+1\right)q^3/{2}+$\\$(q-1)(q^{n+1}+q^n-q^2)/2-1$}  &  \makecell[c]{$[[q^{2n+2}-q^{n+3}+q^{n+2},q^{2n+2}-q^{n+3}+q^{n+2}+$\\$(q-1)(q^{n+1}+q^{n}-q^{2})-2-2s,$\\$\geq s-(q-1)(q^{n+1}+q^n-q^2)+2]]_{q^{2n}}$}      &    GGS curve    &   \cite{BMZ}  \\
        \ &\ &\ &\ \\
        \makecell[c]{$(q-1)(q^n-q)-2\leq s\leq$\\$\left((q^{n-1}-1)\frac{q^n+1}{q+1}+1\right)q^2/{2}+$\\$(q-1)(q^n-q)/2-1$}  &  \makecell[c]{$[[q^{2n+1}-q^{n+2}+q^{n+1},q^{2n+1}-q^{n+2}+q^{n+1}+(q-1)$\\$(q^{n}-q)-2-2s,\geq s-(q-1)(q^{n}-q)+2]]_{q^{2n}}$}      &    Abd$\acute{o}$n-Bezerra-Quoos curve    &   \cite{BMZ}  \\
        \ &\ &\ &\ \\
        $m\mid q+1$,$mq-m-q\leq r\leq m(q-1)-1$  &  $[[mq^2-mq+q, mq^2-m-2r-1, \geq r-mq+m+q+1]]_q$   &    $y^q+y=x^m$    &   Theorem \ref{coro} \\
    \end{tabular}
\end{sidewaystable}

\begin{sidewaystable}[thp]
    \renewcommand{\arraystretch}{2.0}
    \caption{Some known results on self-dual codes from algebraic curves.}
    \centering
    \begin{tabular}{c|c|c|c} 
        \textbf{$q$} & {$n$} & \textbf{algebraic curves} & \textbf{Reference}\\
        \hline
        $q=r^2,r\equiv 3\pmod4$     &      $2tr,t\leq \frac{r-2}{2}$    &  projective line  & \cite{JX} \\
        $4\mid q$  &  $4\leq n\leqq+\lfloor2\sqrt{q}\rfloor-2$   &   $x^3+a_3x^2+a_4x+a_5$     &   \cite{JH}  \\
        $2\mid q$  &  \makecell[c]{$n=2n^{\prime},1\leq n^{\prime}\leq|U|,U=\{\alpha\in\mathbb{F}_{q}|Tr(\alpha^{3}+\alpha)=0\}$}   &   $y^q+y=x^3+x$     &   \cite{Sok}  \\
        $8\mid q$  &  \makecell[c]{$n=2n^{\prime},1\leq n^{\prime}\leq|U|,U=\{\alpha\in\mathbb{F}_{q}|Tr(\alpha^5)=0\}$}   &   $y^q+y=x^5$     &   \cite{Sok}  \\
        \makecell[c]{$q$ odd, $q \equiv 1 \pmod{4}$\\$n \mid (q-1)$, $n$ is odd\\there exists $\beta \subseteq \mathbb{F}_q^*$ such that $\beta^n - 1 \in Q R_q$} & $3n+1$  & hyperelliptic curve  &   \cite{FL} \\
        \ &\ &\ &\ \\
        \makecell[c]{$q = 2^{2e} $, $m\mid (2^e+1)$, $m(n-l)+l$ even\\$U \subseteq \mathbb{F}_q, |U| = n,  U' = \{ a \in U \mid a^{2^e} + a = 0 \}, |U'| = l$} & $m(n-l)+l$  & $y^m=x^{2^e}+x$  &   \cite{FL} \\
        \ &\ &\ &\ \\
        \makecell[c]{$q = 2^{2e} $, $m\mid (2^e+1)$, $m(n-l)+l$ odd\\$U \subseteq \mathbb{F}_q, |U| = n,  U' = \{ a \in U \mid a^{2^e} + a = 0 \}, |U'| = l$} & $m(n-l)+l+1$  & $y^m=x^{2^e}+x$  &   \cite{FL} \\
        \ &\ &\ &\ \\
        \makecell[c]{$q = q_0^2$ for odd $q_0$, $m \mid (q_0 + 1), n = \mu(q_0 - 1) + \nu(q_0 + 1)$\\$1 \leq \mu < \frac{q_0 + 1}{2}$, $ 1 \leq \nu < \frac{q_0 - 1}{2}$\\ $q\equiv 1(\mod4),\mu$ even, $m(n-\nu)+\nu$ even} & $m(n-\nu)+\nu$  & $y^m=x^{2^e}+x$  &   \cite{FL} \\
        \ &\ &\ &\ \\
        \makecell[c]{$q = q_0^2$ for odd $q_0$, $m \mid (q_0 + 1), n = \mu(q_0 - 1) + \nu(q_0 + 1)$\\$1 \leq \mu < \frac{q_0 + 1}{2}$, $ 1 \leq \nu < \frac{q_0 - 1}{2}$\\ $q\equiv 3(\mod4),\mu$ odd, $m(n-\nu)+\nu$ odd} & $mn$  & $y^m=x^{2^e}+x$  &   \cite{FL} \\
        \ &\ &\ &\ \\
        $q = p^{2e}$ for odd $ p $, $ m \mid (q_0 + 1) $, $ 1 \leq c < e $ & $m(p^{2c}-1)+2$  & $y^m=x^{2^e}+x$  &   \cite{FL} \\
    \end{tabular}
\end{sidewaystable}

\begin{sidewaystable}[thp]
    \renewcommand{\arraystretch}{2.0}
    \caption*{Table 2: Some known results on self-dual codes from algebraic curves.} 
    \centering
    \begin{tabular}{c|c|c|c} 
        \textbf{$q$} & {$n$} & \textbf{algebraic curves} & \textbf{Reference}\\
        \hline
        \makecell[c]{$q = q_0^2$, $ q_0 \equiv3(\mod4)$, $ m \mid (q_0 + 1) $, $ s,t | (q - 1) $ \\
        $ 0 \leq \mu < \frac{s}{gcd(s,t)} $, $ 0 \leq \nu < \frac{t}{gcd(s,t)} $ \\
        $I = \{ c =v, v+1, \dots, \frac{q-1}{s}-1\mid \frac{2cs}{q_0+1}\equiv 1(\mod2) \}$, $|I|=l$ \\
        $s$ even, $ s\mid(q_0 + 1), t\mid(q_0 - 1), \gcd(s, \frac{q_0 + 1}{2}) > 2 $    \\
        or $(\frac{q-1}{s}-1)t-\frac{t(q_0+1)}{s}\mu$ and $\frac{t(q_0+1)}{s}\mu+t$ even    \\
        $m(\mu\frac{q-1}{s} + \nu\frac{q-1}{t}-2\mu\nu\frac{(q-1)gcd(s,t)}{st}-l)+l$ even
        } & $m(\mu\frac{q-1}{s} + \nu\frac{q-1}{t}-2\mu\nu\frac{(q-1)gcd(s,t)}{st}-l)+l$  & $y^m=x^{2^e}+x$  &   \cite{FL} \\
        \ &\ &\ &\ \\
        \makecell[c]{$q = q_0^2$, $ q_0 \equiv3(\mod4)$, $ m \mid (q_0 + 1) $, $ s,t | (q - 1) $ \\
        $ 0 \leq \mu < \frac{s}{gcd(s,t)} $, $ 0 \leq \nu < \frac{t}{gcd(s,t)} $ \\
        $I = \{ c =v, v+1, \dots, \frac{q-1}{s}-1\mid \frac{2cs}{q_0+1}\equiv 1(\mod2) \}$, $|I|=l$ \\
        $s$ even, $ s\mid(q_0 + 1), t\mid(q_0 - 1), \gcd(s, \frac{q_0 + 1}{2}) > 2 $    \\
        or $(\frac{q-1}{s}-1)t-\frac{t(q_0+1)}{s}\mu$ and $\frac{t(q_0+1)}{s}\mu+t$ even    \\
        $m(\mu\frac{q-1}{s} + \nu\frac{q-1}{t}-2\mu\nu\frac{(q-1)gcd(s,t)}{st}-l)+l$ odd
        } & $m(\mu\frac{q-1}{s} + \nu\frac{q-1}{t}-2\mu\nu\frac{(q-1)gcd(s,t)}{st}-l)+l+1$  & $y^m=x^{2^e}+x$  &   \cite{FL} \\
        \ &\ &\ &\ \\
        $q$ even, $m \mid q+1$  &  $n=mq^2-mq+q$   &    $y^q+y=x^m$    &  Theorem \ref{mc}  \\
        $q$ even or $s$ odd, $s \mid q^2-1$  &  $n=q(s+1)$   &    Hermitian curve    &  Theorem \ref{q(q+1)}  \\
        $q$ even or $k$ odd, $q+1\mid k$  &  $n=q^3-qk$   &    Hermitian curve    &   Theorem \ref{T10}  \\
        $q=p^t$ even, $k\leq 2t$    &  $n=p^kq$   &    Hermitian curve    &   Theorem \ref{T11}  \\
        \ &\ &\ &\ \\
        $q=p^m$ even, $s\mid 2m$, $2t\leq p^s$   &  $n=q(2t+1)p^{sl}$   &    Hermitian curve    &   Theorem \ref{T6}  \\
        \ &\ &\ &\ \\
        $q=p^m$, $s\mid 2m$, $2t\leq p^s$  &  $n=2tqp^{sl}$   &    Hermitian curve    &   Theorem \ref{T7}  \\
    \end{tabular}
\end{sidewaystable}

\begin{table}[thp]
    \renewcommand{\arraystretch}{2.0}
    \caption{Our results on self-orthogonal codes from algebraic curves.}
    \centering
    \begin{tabular}{c|c|c|c|c} 
        \textbf{Parameters restrictions} & $n$ & $k$ & $d$ & \textbf{Theorem} \\
        \hline
        \makecell[c]{$2g-1\leq r\leq g-1+\frac{q(2t+1)p^{sl}}{2}$\\$q = p^m, s \mid 2m, 2t \leq p^s, sl\leq m$\\$g=\frac{q(q-1)}{2}$}   &  $q(2t+1)p^{sl}$   &   $r+1-\frac{q(q-1)}{2}$     &   $\geq q(2t+1)p^{sl}-r$ & Theorem \ref{T6}  \\

        $mq-m-q\leq r\leq\lfloor\frac{m(q^2-1)-1}{2}\rfloor$     &      $mq^2-mq+q$    &  $r+1-\frac{q(q-1)}{2}$  & $\geq q(2t+1)p^{sl}-r$ & Theorem \ref{mc} \\

        $q^2-q-1\leq r\leq\lfloor\frac{1}{2}q(q+s)-1\rfloor$  &  $q(s+1)$   &   $r-\frac{1}{2}q(q-1)+1$     &   $\geq q(s+1)-r$ & Theorem \ref{q(q+1)}  \\

        $2g-1\leq r\leq g-1+\lfloor \frac{q^3-qk}{2}\rfloor$  &  $q^3-qk$   &    $r-\frac{q(q-1)}{2}+1$    &   $\geq q^3+\frac{q(q-1)}{2}-qk-1$ & Theorem \ref{T10}  \\

        $q^2-q-1\leq r\leq\lfloor\frac{q(p^k+q-1)}{2}-1\rfloor$.  &  $p^kq$   &    $r-\frac{1}{2}q(q-1)+1$    &   $\geq p^kq-r$ & Theorem \ref{T11}  \\
    \end{tabular}
\end{table}

\end{document}